\documentclass[10pt,a4paper]{article}
\usepackage[utf8]{inputenc}
\usepackage[english]{babel}
\usepackage{hyperref}
\usepackage{url}
\usepackage{amsmath,amsthm,enumerate}
\usepackage{amssymb}
\usepackage{tikz}
\usetikzlibrary{decorations.pathreplacing}
\usepackage{subfigure}
\usepackage[hmargin=2.5cm,vmargin=3cm]{geometry}

\usepackage{algorithm}
\usepackage{algorithmic}

\newtheorem{theorem}{Theorem}

\newtheorem{proposition}[theorem]{Proposition}
\newtheorem{lemma}[theorem]{Lemma}

\newcommand{\decisionpb}[4]{
  
        \begin{minipage}{#4\textwidth}
                #1\\
                \emph{Instance:} #2\\ 
                \emph{Question:} #3
        \end{minipage}
\vspace{0.2cm}
}

\newcommand{\id}{\gamma^{\text{\tiny{ID}}}}

\definecolor{bleujoli}{RGB}{23.,74.,124.}
\definecolor{rosejoli}{RGB}{176.,0,64.}
\definecolor{vertjoli}{RGB}{34, 120, 15}
\definecolor{rougejoli}{RGB}{217.,1,21.}
\definecolor{orang}{RGB}{255,70,0}

\tikzstyle{noeud}=[circle,inner sep=2, minimum size =2 pt, line width = 1pt, draw=black, fill=white]
\tikzstyle{noeud_ext}=[circle,inner sep=3, minimum size =6 pt, line width = 1pt, draw=black, fill=bleujoli]

\tikzstyle{cercle}=[line width = 1.5pt, draw=rougejoli]

\newcommand{\Dabcd}[4]{\mathcal{D}^{[#3,#4]}_{[#1,#2]}}

\begin{document}

\title{Identification of points using disks}
\author{Valentin Gledel\footnote{\noindent Univ Lyon, Université Claude Bernard Lyon 1, LIRIS - CNRS UMR 5205, F69622 (France). E-mail: valentin.gledel@univ-lyon1.fr}
\and Aline Parreau\footnote{\noindent Univ Lyon, Université Claude Bernard Lyon 1, CNRS, LIRIS - CNRS UMR 5205, F69622 (France). E-mail: aline.parreau@univ-lyon1.fr}}

\maketitle

\begin{abstract}
  We consider the problem of identifying $n$ points in the plane using disks, i.e., minimizing the number of disks so that each point is contained in a disk and no two points are in exactly the same set of disks. This problem can be seen as an instance of the {\em test covering problem} with geometric constraints on the tests. We give tight lower and upper bounds on the number of disks needed to identify any set of $n$ points of the plane. In particular, we prove that if there are no three colinear points nor four cocyclic points, then $2 \lceil n/6 \rceil + 1$ disks are enough, improving the known bound of $\lceil (n+1)/2 \rceil$ when we only require that no three points are colinear.
  We also consider complexity issues when the radius of the disks is fixed, proving that this problem is NP-complete. In contrast, we give a linear-time algorithm computing the exact number of disks if the points are colinear.
\end{abstract}


\section{Introduction}

Let $\mathcal P$ be a set of $n$ points of the plane $\mathbb R^2$. What is the minimum number of disks so that each point is contained in a disk and no two points are in exactly the same set of disks? In other words, we want to find a minimum set of disks such that every point is in a disk and the disks that contain a given point uniquely determine it. Such a set of disks (not necessarily minimum) is said to {\em identify} $\mathcal P$. See Figure~\ref{fig:ex} for an example of an identifying set of disks.

\begin{figure}[ht]
  \begin{center}
  \begin{tikzpicture}

\foreach \pos in {(0,0),(-0.5,1.5),(1,-0.3),(1,2),(1.2,0.7),(1.9,0.3),(1.7,1.8),(2.3,1)}{
\node[noeud] at \pos {};}

\draw[cercle] (-0.1,0.5) circle (1.65cm);
\draw[cercle] (0.6,1.65) circle (1.3cm);
\draw[cercle] (2.15,0.5) circle (1.6cm);
\draw[cercle] (1.3,0.05) circle (1cm);

\end{tikzpicture}
  \end{center}
  \label{fig:ex}
  \caption{A set of four disks optimally identifying eight points.}
  \end{figure}
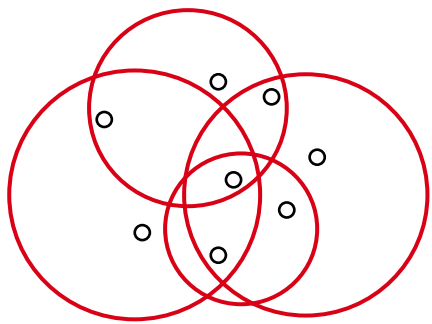

The motivation of this problem comes from the localization of indviduals and more generally from the following setting of identification problems: Given a set of individuals with binary attributes that each individual can have or not, the goal is to choose a minimum number of attributes in such a way that each individual has a unique set of attributes. This problem is known in the literature as the {\em test covering problem} \cite{MS85} or {\em identifying codes problem in hypergraph} \cite{PhDMoncel} since one can represent the data by a hypergraph where individuals are vertices and attributes are hyperedges. It has many application in particular in medical diagnostics and pattern recognition \cite{K95,MS85,PP80}.

In a context of localization, the attributes are defined by the metric of the space where individuals live. As an example, in the context of {\em identifying codes} \cite{KCL98}, individuals are vertices of a graph. Then the attributes are defined by the closed neighbourhoods, meaning ``to be closed to''. Choosing some attributes is equivalent to setting detectors on some vertices that are able to detect errors in their neighbourhood. Then the set of detectors is able to detect any intrusion in the graph. Indeed, if there is an intrusion on a vertex, then the set of detectors that have detected something uniquely determines where the intrusion is. Locating-dominating sets \cite{S88,SR84} and open locating dominating sets \cite{SS11} are defined in a similar way. These concepts are studied by various authors since the 1970s and 
1980s, and have been applied to various regions such as fault-detection in networks~\cite{KCL98,UTS04} or graph isomorphism testing~\cite{B80}.

In this paper, we consider that individuals (which are the points in our problem) are living in $\mathbb R^2$. A detector can be placed anywhere, with any radius of detection and thus is represented by a disk. It can be formulated as a test covering instance: a set of individuals share an attribute if they can be isolated from the other individuals by a disk. It is also related to identifying codes in graphs: if the detectors must be located on points and have a fixed radius, a natural graph structure emerges. Then our problem is equivalent to the problem of identifying codes in unit disk graphs (in the general case) or in unit interval graphs (if points are colinear).

Another motivation comes from the notion of {\em geometric separator} in computational geometry \cite{DHMS01}. Let $C_1,\ldots,C_k$ be $k$ finite disjoint sets of $\mathbb R^2$. A finite set $S$ of curves in the plane is a {\em separator} for the sets $C_1$,...,$C_k$ if every connected component in $\mathbb R^2-S$ contains points from only one set $C_i$. Finding separators is a classical problem of computationnal geometry, in particular when considering image analysis. The most studied case is $k=2$ and separation with lines or circles \cite{AHMSS00}. Our problem -- if we forget the condition that each point must be in a disk -- can be considered as a separating problem where each set $C_i$ contains only one point and $S$ is a union of circles. This problem has been mentionned by Gerbener and T\'oth \cite{GT13} who have considered more generally separation with convex sets. They in particular proved that $\lceil n/2\rceil$ circles are enough to separate $n$ points even if they are in a general configuration (no three colinear points). Separators of single points have also been studied for lines. Bolland and Urrutia \cite{BU95} gave an algorithm of time complexity $O(n\log n)$ to find a family of $n/2$ lines that separates any set of $n$ points in a general configuration. C\v{a}linescu, Dumitrescu and Wan \cite{CDW05} proved that in the particular case where the lines are parallel to the axis, the problem is NP-complete and gave a constant approximation polynomial algorithm for this case. A natural extension in higher dimension, called {\em multi-modal sensor allocation problem}, has been defined in \cite{KSPS02}, making links with identification problems.
Note that the separating problem with lines is a subproblem of ours. Indeed, if the points are given, one can consider a line as a very large circle.

In Section \ref{sec:defi}, we give formal definitions and background that will need along the paper. In Section~\ref{sec:caspart}, we study particular configurations of points: colinear or forming a grid. For colinear points, we give the exact number of disks needed if any radius can be used. If the points are on a grid, we give exact values for height 2 and bounds for larger heights. In Section \ref{sec:bounds}, we give tight lower and upper bounds: we prove that at least $\Theta(\sqrt{n})$ and at most $\lceil (n+1)/ 2 \rceil$ disks are necessary ($n$ is the number of points). If moreover there are no three colinear points nor four cocyclic points, then we prove, using Delaunay triangulation, that we need at most $2\lceil n/6 \rceil +1$ disks. Finally, in Section \ref{sec:complexity}, we discuss the complexity of the problem when the radius is fixed , we prove that it is NP-complete in the general case but that there is a linear algorithm to solve it when the points are colinear.

\section{Definition and background}\label{sec:defi}

\subsection{Formal definition}

Let $\mathcal P$ be a set of points of $\mathbb R^2$. A disk of radius $r\in \mathbb R$ and center $c\in \mathbb R^2$ is the set of points of $\mathbb R^2$ at distance at most $r$ of $c$. A point $P\in \mathcal P$ is {\em covered by} a disk if it belongs to it. Two points $P$ and $Q$ of $\mathcal P$ are {\em separated} by a disk $D$ if exactly one of them is covered by $D$. A set of disks $\mathcal D$ is {\em identifying} $\mathcal P$ if it is covering all the points of $\mathcal P$ and separating all the pairs of points of $\mathcal P$. We denote by $\id_D(\mathcal P)$ the minimum number of disks needed to identify $\mathcal P$. Let $r\in \mathbb R$, we denote by $\id_{D,r}(\mathcal P)$ the minimum number of disks of radius $r$ needed to identify $\mathcal P$. When $r$ is large enough compare to the distances between the points of $\mathcal P$, any disk of radius $r$ is separating the same pairs of points as some half-plane. Hence, identification with half-planes is a particular case of identification with disks of fixed size. We will denote by $\id_{D,\infty}(\mathcal P)$ the corresponding number. 

\paragraph{Remark.}
In our definition, we ask that every point of $\mathcal P$ must be covered by at least one disk. This choice could be discussed. Indeed, it is not the case for similar notions like {\em separating families} or {\em test covers}. We choose this definition to be consistent with our first motivation: in a context of localization, our detection system must be able to detect if there is an intrusion or not, which is possible only if all the points are covered. This is the reason why in {\em identifying codes} there is the condition of domination (see Section \ref{sec:preli} for formal definition of identifying codes).
However, if a set of disks is separating all the pairs of points of $\mathcal P$, at most one point is not covered (otherwise all the points that are not covered will not be separated). Therefore, we need at most one more disk to obtain an identifying set. Hence the difference between the two values is at most one and our results can be easily adapted for only-separating sets.

  \

For any radius $r$ and any points $P$ and $Q$ of $\mathbb R^2$, there is always a disk of radius $r$ that separates them. Hence, $\id_D(\mathcal P)$ and $\id_{D,r}(\mathcal P)$ are well-defined and always smaller than ${\binom{|\mathcal P|}{2}}$. Moreover, if the radius is not fixed or small enough, one can take for each point a disk containing only this point and then forms an identifying set of disks. Thus, we have $\id_D(\mathcal P)\leq |\mathcal P|$. About the lower bound, consider a set $\mathcal D$  of $k$ disks identifying $\mathcal P$. Since each point is contained in a unique non-empty subset of $\mathcal D$, there are at most $2^k-1$ points in $\mathcal P$, leading to the following lower bound on $\id_{D}(\mathcal{P})$:

\begin{lemma}
Let $\mathcal{P}$ be a set of $n$ points of $\mathbb R^2$, then
$\id_{D}(\mathcal P) \geq  \lceil \log(n+1) \rceil$.
\label{lem:logbound}
\end{lemma}

These trivial lower and upper bounds are not tight and will be improved in Section \ref{sec:bounds}.

Finally, since a set of disks identifying $\mathcal P$ is identifying any subset of $\mathcal P$, we have the following lemma:

\begin{lemma}
Let $\mathcal{P}$ and $\mathcal{P'}$ be two sets of points of $\mathbb R^2$ with  $\mathcal{P}' \subseteq \mathcal{P}$, then
$\id_{D}(\mathcal P) \geq  \id_{D}(\mathcal{P}')$.
\label{lem:inclus}
\end{lemma}

\subsection{Related work}\label{sec:preli}

Among the related notions given in the introduction, we give formal definitions for three of them that we will need in the rest of the paper.

\paragraph{Separating families of disks.}

If $\mathcal D$ is only separating any pair of vertices of $\mathcal P$, $\mathcal D$ is a {\em separating family of disks}, studied by Gerbner and T\'oth \cite{GT13} in the more general context of convex sets. They in particular consider the parameter $s(n,\mathcal D)$ and $s'(n,\mathcal D)$) which stand for the maximum number of disks that are needed to separate any $n$-point set and any $n$-point set in general position (no three of its points are on a line). They prove that $s(n,\mathcal D)=s'(n,\mathcal D)=\lceil n/2 \rceil$. Since at most one more disk is necessary to obtain an identifying set of disks from a separating set of disks, it means that $\id_{D}(\mathcal P)$ is at most $\lceil n/2 \rceil+1$. We will improve this bound in Section \ref{sec:bounds} to $2\lceil n/6 \rceil +1$ if moreover no four points are cocyclic.

 \paragraph{Identifying codes in unit interval and unit disk graphs.}

Let $G=(V,E)$ be a graph.  A vertex $c$ {\em dominates} a vertex $x$ if $c$ is in the closed neighbourhood of $x$ (i.e: $x$ and its neighbours). It {\em separates} two vertices $x$ and $y$ if it is dominating exactly one of them. An {\em identifying code} of $G$ is a subset of vertices $C$ such that each vertex is dominated by some vertex of $C$ and each pair of vertices of $G$ is separated by some vertex of $C$. We denote by $\id(G)$ the minimum number of vertices in an identifying code of $G$. Note that $\id(G)$ is not always well-defined since $G$ might have two vertices with exactly the same neighbourhood and thus no vertex can separate them. Such vertices are called {\em twin vertices}. If a graph does not have any twins, then it has an identifying code (take for example all the vertices in $C$).

Identifying codes are closely related to identifying sets of disks when considering graphs of geometric intersections.
Given a set of geometric objects, one can define its intersection graph as follows. Vertices are the objects and there is an edge between two objects if they intersect. A class of graphs of particular interest for us is the class of {\em unit disk graphs} that are the intersection graphs of disks of radius 1.
Let $G$ be a unit disk graph and denote by $\mathcal P$ the set of centers of the disks forming $G$. Then an identifying code of $G$ is equivalent to an identifying set of $\mathcal P$ using disks that have radius 2 and are centered on points of $\mathcal P$. Indeed, a disk of radius 2 centered on a point $P$ of $\mathcal P$  contains all the points that are centers of disks of the closed neighbourhood of the disk corresponding to $P$ in $G$.
Identifying codes in unit disk graphs have been studied by M\"uller and Sereni \cite{MS09} who prove, in particular, that the minimization problem in NP-complete. If the points of $\mathcal P$ are colinear, then $G$ is a {\em unit interval graph}. The complexity of identifying codes in unit interval graphs is surprisingly still open \cite{FMNPV} (but has been proved to be NP-complete for interval graphs).

Junnila and Laihonen \cite{JL11} studied identifying codes in the grid $\mathbb Z^2$ using Euclidean balls. The underlying graph has the set $\mathbb Z^2$ as vertices and the closed neighbourhood are given by the Euclidean balls of a fixed radius $r$. This graph can also be seen as an (infinite) unit disk graph. They give lower and upper bounds  on the density of minimum identifying codes in this graph in function of $r$.

\paragraph{Identifying codes in hypergraphs.}
The notion of identifying codes can be extended to hypergraphs. Let $\mathcal H=(V,\mathcal E)$ be a hypergraph. An {\em identifying code} of $\mathcal H$ is a set $C\subseteq \mathcal E$ of hyperedges such that:
\begin{itemize}
\item each vertex of $\mathcal H$ is in at least one element of $C$;
  \item for each pair of vertices of $\mathcal H$, there is an element of $C$ containing exactly one element of the pair.
\end{itemize}

An identifying code in a graph $G$ is equivalent to an identifying code in the hypergraph of the closed neighbourhoods of $G$. As said in the introduction, this notion is known under different names and has actually been introduced before identifying codes in graphs, see \cite{PhDMoncel, MS85}.
Our problem can be reduced to identifying codes in hypergraph. Indeed, let $\mathcal P$ be a set of $n$ points of $\mathbb R^2$. Let $\mathcal H (\mathcal P)$ be the hypergraph with vertex set $\mathcal P$ and a set of points $E\subseteq \mathcal P$ is a hyperdege if there exists a disk $D$ such that $D\cap \mathcal P=E$. Then finding an identifying set of disks identifying $\mathcal P$ is equivalent to finding an identifying code in $\mathcal H(\mathcal P)$. Note that an hyperedge of $\mathcal H(\mathcal P)$ of size $k$ corresponds to a nonempty cell in the iterated Vorono\"i diagram of size $k$ of $\mathcal P$ and can be computed in $O(n)$ time \cite{orourke}. The whole hypergraph $\mathcal H(\mathcal P)$ can be obtained by computing all iterated Vorono\"i diagrams of $\mathcal P$. This can be done in time $O(n^3)$ and the number of hyperedges of $\mathcal H(\mathcal P)$ is of order $O(n^3)$ \cite{orourke2}.

\section{Particular configurations} \label{sec:caspart}

\subsection{Colinear points}

When points are located on a single line, the problem is completly solved with the following theorem.
\begin{theorem}
\label{th:align}
Let $\mathcal{P}$ be a set of $n$ colinear points, then $\id_{D}(\mathcal P) = \lceil \frac{n+1}{2} \rceil$.
\end{theorem}

\begin{proof}
  Let $\mathcal{P}$ be a set of $n$ colinear points located on a line $L$. We denote by $x_1,...,x_n$ the points, respecting their order on $L$.

  Let $\mathcal D$ be a set of disks identifying $\mathcal P$.
  For any $i\in\{1,...,n-1\}$, $x_i$ and $x_{i+1}$ are separated by $\mathcal D$. It means that there is a disk $D\in \mathcal D$, such that its perimeter intersects $L$ between $x_i$ and $x_{i+1}$. Moreover, $x_1$ and $x_n$ are covered by $\mathcal D$, thus there is a disk whose perimeter intersects $L$ before $x_1$ and after $x_n$. In total, there are at least $n+1$ intersections between $L$ and some disks' perimeters. Since a circle intersects a line into at most two points, we necessarily have $|\mathcal D|\geq \lceil \frac{n+1}{2}\rceil$.

To prove the equality, note that for any subset of consecutives points $x_i,x_{i+1},...,x_j$ of $\mathcal P$, there exists a disk $D_{i,j}$ such that $D_{i,j}\cap \mathcal P=\{x_i,x_{i+1},...,x_j\}$. Then the set of disks $$\mathcal D=\left\{D_{i,i+\lceil n/2 \rceil}\  |\  i=1,..,\left \lceil \frac{n+1}{2} \right \rceil \right\}$$has size $\lceil \frac{n+1}{2}\rceil$ and is identifying $\mathcal P$. See Figure \ref{fig:algo_align} for an illustration with nine points.

\end{proof}

\begin{figure}[ht]

\begin{center}
\scalebox{0.8}{
\begin{tikzpicture}

\foreach \pos in {(0,0),(1,0),(3,0),(3.5,0),(5,0),(6,0),(7.5,0),(8.5,0),(9,0)}{
\node[noeud] at \pos {};}

\draw (0,0) node[above=2pt]{$x_1$};
\draw (1,0) node[above=2pt]{$x_2$};
\draw (3,0) node[above=2pt]{$x_3$};
\draw (3.5,0) node[above=2pt]{$x_4$};
\draw (5,0) node[above=2pt]{$x_5$};
\draw (6,0) node[above=2pt]{$x_6$};
\draw (7.5,0) node[above=2pt]{$x_7$};
\draw (8.5,0) node[above=2pt]{$x_8$};
\draw (9,0) node[above=2pt]{$x_9$};

\draw (2.5,0) circle (3cm);
\draw (3.5,0) circle (3cm);
\draw (5.25,0) circle (2.75cm);
\draw (6,0) circle (2.75cm);
\draw (7,0) circle (2.5cm);
\end{tikzpicture}
}
\end{center}
\caption{Identification of colinear points}
\label{fig:algo_align}
\end{figure}
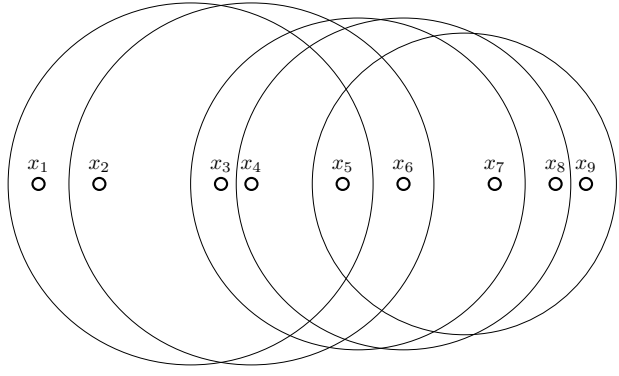

In the solution given in the proof of Theorem \ref{th:align}, some disks might have a big radius. Actually, if the radius of the disks is bounded by a constant $r$, $n$ disks are sometimes needed and $\id_{D,r}(\mathcal P)$ can take any value between $\lceil \frac{n+1}{2} \rceil$ and $n$. In Section \ref{sec:complexity}, we give an algorithm that computes  $\id_{D,r}(\mathcal P)$ in linear time.

\subsection{Points located on a grid}

We now consider points located on a regular grid. Given two integers $m$ and $n$, we denote by $\mathcal P_{m,n}$ the set of points $(x,y)$ of $\mathbb Z^2$ such that $1\leq y\leq m$ and $1\leq x \leq n$.

\subsubsection{Grids of height 2}

When the grid contains only two lines, one can identify the points using the same number of disks than on a single line, except in few cases:

\begin{theorem}
  Let $n\geq 2$ be an integer. We have:
  \begin{equation*}
  \id_{D}(\mathcal P_{2,n})=\begin{cases}
    \lceil \frac{n+1}{2} \rceil+1 & \text{if} \  n\in\{2,3,4,5,7\},\\
    \lceil \frac{n+1}{2} \rceil & \text{otherwise}.
  \end{cases}
\end{equation*}
\end{theorem}

\begin{proof}

We can first see that for all $n$, $\mathcal P_{2,n} \leq \lceil \frac{n+1}{2} \rceil+1$. Indeed, to identify $\mathcal{P}_{2,n}$ one can use the method proposed in Theorem~\ref{th:align} and add an half-plane (which can be seen as a very large disk) to separate the lines as in Figure~\ref{fig:2lines}.

\begin{figure}[ht]
\begin{center}
\scalebox{0.8}{
\begin{tikzpicture}
\foreach \pos in {(0,0),(1,0),(2,0),(3,0),(4,0),(5,0),(6,0),(0,1),(1,1),(2,1),(3,1),(4,1),(5,1),(6,1)}{
\node[noeud] at \pos {};}
\draw (1.5,0.5) circle (2cm);
\draw (2.5,0.5) circle (2cm);
\draw (3.5,0.5) circle (2cm);
\draw (4.5,0.5) circle (2cm);
\draw (-1,0.5) -- (7,0.5);
\end{tikzpicture}
}
\end{center}
\caption{Identification of $\mathcal{P}_{2,n}$ with $\lceil \frac{n+1}{2} \rceil+1$ disks}
\label{fig:2lines}
\end{figure}
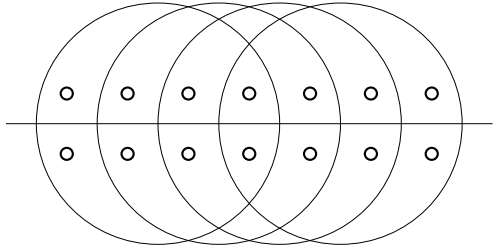

For grids $\mathcal{P}_{2,n}$ with $n\leq 5$, this solution is optimal by Lemma~\ref{lem:logbound}. In Section~\ref{sec:bounds} Proposition~\ref{prop:lowerbound}, we show that at least five disks are needed to identify a set of 14 points, so there is no better solution for $\mathcal{P}_{2,7}$. For all the other cases, we show that we only need $\lceil\frac{n+1}{2}\rceil$ disks.
We only have to study the case where $n$ is odd or equal to 6. Indeed, by Lemma~\ref{lem:inclus}, solution for $P_{2,2q+1}$ is also a solution for $P_{2,2q}$ by removing the points of the last column.

We first give a characterization for a set $X \subseteq \mathcal{P}_{2,n}$ to be the intersection of $\mathcal{P}_{2,n}$ and a disk.
Let $X$ be such a set, $X$ is the union of two sets of consecutive points of the first line $(a,1)$, ... , $(b,1)$ and of the second line $(c,2)$, ... , $(d,2)$, with $a,b,c,d \in \mathbb{N}$, $a \leq b$ and $c \leq d$. We must have either $[a,b] \subseteq [c,d]$ or $[c,d] \subseteq [a,b]$ and the difference between each extremities must differ of at most 1: $|(c-a)-(b-d)| \leq 1$. 
This condition is sufficient since for every $a,b,c,d$ verifying this relation, there exist a disk $\Dabcd{a}{b}{c}{d}$ that contains exactly these consecutive points.

\medskip

An explicit solution for the grids $P_{2,6}$, and $P_{2,9}$ are the following disks :

\begin{itemize}
\item
$P_{2,6}$ can be identified by the set of disks : $\Dabcd{1}{5}{3}{4}$, $\Dabcd{3}{6}{4}{5}$, $\Dabcd{2}{3}{1}{5}$ and $\Dabcd{4}{4}{2}{6}$.

\item  $P_{2,9}$ can be identified by the set of disks : $\Dabcd{1}{6}{3}{4}$, $\Dabcd{2}{9}{4}{6}$, $\Dabcd{4}{8}{6}{7}$, $\Dabcd{3}{4}{1}{8}$ and $\Dabcd{6}{7}{2}{9}$.
\end{itemize}

We now give a solution for grids $\mathcal{P}_{2,4p+1}$, with $p\geq 3$.
This solution use three different steps.
Figure~\ref{fig:grids_2*n} gives an illustration of these three steps.

The first step is to use the disks $\mathcal{D}_1 = \Dabcd{1}{3p+1}{p+2}{2p}$, $\mathcal{D}_2 = \Dabcd{p+2}{2p}{1}{3p+1}$, $\mathcal{D}_3 = \Dabcd{p+1}{4p+1}{2p+2}{3p}$ and $\mathcal{D}_4~=~\Dabcd{2p+2}{3p}{p+1}{4p+1}$.
After adding these disks, the points of each line are separated from the other line. Indeed, the points of the first line in the intervals $[1,p+1]$ and $[2p+1,3p+1]$ are in the disk $\mathcal{D}_1$ and are not in the disk $\mathcal{D}_2$ which separate them from all the points of the second line. Similarly the points of the first line and in the intervals $[p+2,2p]$ and $[3p+2,4p+1]$ are in the disk $\mathcal{D}_3$ and are not in the disk $\mathcal{D}_4$, which separate them from all the points of the second line.

In the second step, we add the disks $\Dabcd{p}{p+2}{p}{p+2}$ and $\Dabcd{3p}{3p+2}{3p}{3p+2}$. These two disks separate the points on the columns $p+1$, $2p+1$ and $3p+1$, which weren't until now.

After this, all the points are covered by at least one disk and the points that are no separated from each other are the same line and on the intervals $[1,p-1]$, $[3p+3,4p+1]$, $[p+3,2p]$ and $[2p+2,3p-1]$ (the last two intervals occurs if $p\geq 4$).

In the third step, we can now finish identifying the points by adding the following concentric disks :

$\Dabcd{2}{4p}{2}{4p}$, $\Dabcd{3}{4p-1}{3}{4p-1}$, ... , $\Dabcd{p-1}{3p+3}{p-1}{3p+3}$, $\Dabcd{p+4}{3p-2}{p+4}{3p+2}$, $\Dabcd{p+5}{3p-3}{p+5}{3p-3}$, ... , $\Dabcd{2p}{2p+2}{2p}{2p+2}$.

We use four disks in the first part, then two disks and finally $(p-2)+(p-3)$. So in total we use $2p+1$ disks, which is equal to $\frac{(4p+1)+1}{2}$ disks.

For the grids $\mathcal{P}_{2,4p-1}$, we can remove the points of the columns 1 and $4p+1$ and the disk $\Dabcd{2}{4p}{2}{4p}$.

So we can indeed identify the grids $\mathcal{P}_{2,n}$, with $n\geq 10$ with $\frac{n+1}{2}$ disks.

\end{proof}

\begin{figure}[ht]
\centering

\scalebox{0.8}{

\begin{tikzpicture}

\clip (0,0) rectangle (18,3);

\draw (7,-16) circle (18.03);
\draw (7,19) circle (18.03);
\draw (11,-16) circle (18.03);
\draw (11,19) circle (18.03);

\draw (0.25,1) node {$\mathcal{D}_1$};
\draw (0.25,2) node {$\mathcal{D}_2$};
\draw (17.75,1) node {$\mathcal{D}_3$};
\draw (17.75,2) node {$\mathcal{D}_4$};

\draw (1,0.5) node {$1$};
\draw (5,0.5) node {$p+1$};
\draw (9,0.5) node {$2p+1$};
\draw (13,0.5) node {$3p+1$};
\draw (17,0.5) node {$4p+1$};

\foreach \x in {1,...,17}
{
	\foreach \y in {1,2}
	{
		\node[noeud] at (\x,\y) {};
	}
}

\end{tikzpicture}
}

\vspace{0.25cm}

\scalebox{0.8}{
\begin{tikzpicture}

\clip (0,0) rectangle (18,3);

\draw (7,-16) circle (18.03);
\draw (7,19) circle (18.03);
\draw (11,-16) circle (18.03);
\draw (11,19) circle (18.03);

\draw (5,1.5) circle (1.12);
\draw (13,1.5) circle (1.12);

\foreach \x in {1,...,17}
{
	\foreach \y in {1,2}
	{
		\node[noeud] at (\x,\y) {};
	}
}

\end{tikzpicture}
}

\vspace{0.25cm}

\scalebox{0.8}{

\begin{tikzpicture}

\clip (0,0) rectangle (18,3);

\draw (7,-16) circle (18.03);
\draw (7,19) circle (18.03);
\draw (11,-16) circle (18.03);
\draw (11,19) circle (18.03);

\draw (5,1.5) circle (1.12);
\draw (13,1.5) circle (1.12);

\draw (9,1.5) circle (1.12);
\draw (9,1.5) circle (6.03);
\draw (9,1.5) circle (7.02);

\foreach \x in {1,...,17}
{
	\foreach \y in {1,2}
	{
		\node[noeud] at (\x,\y) {};
	}
}

\end{tikzpicture}
}
\caption{The three steps of the proof the grid $\mathcal{P}_{2,4p+1}$}
\label{fig:grids_2*n}
\end{figure}
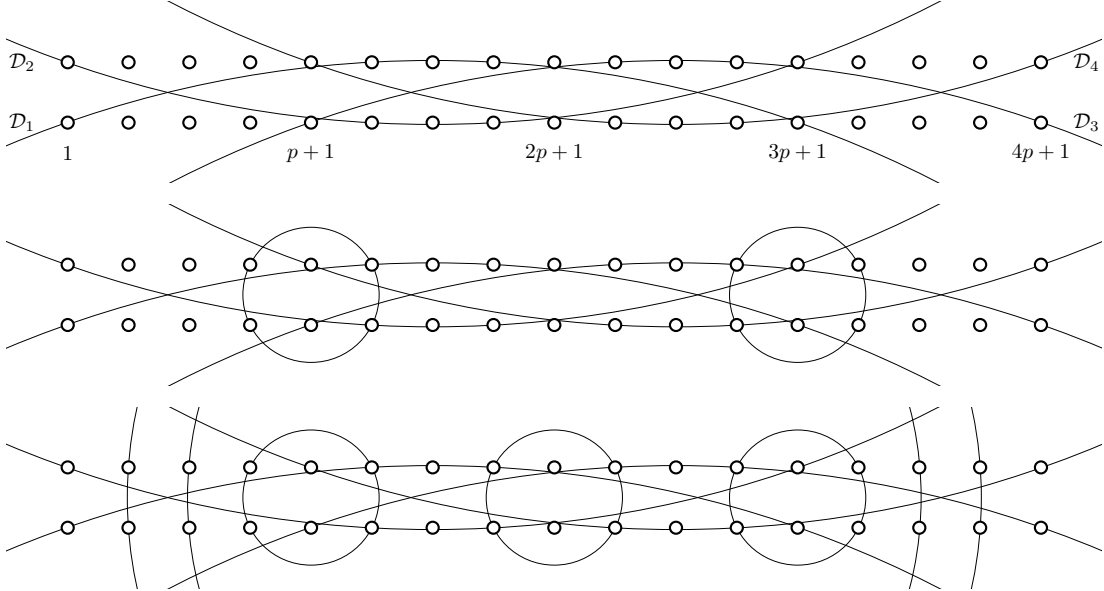

\subsubsection{General case}

We now consider the general case of grids $m \times n$, $n\geq m\geq 3$. We first solve the case of identification with half-planes - which can be considered as disks with infinite radius.

\begin{theorem}
Let $m,n\geq 3$ be two integers. Then, $\id_{D,\infty}(\mathcal P_{m,n}) = m+n-2$.
\end{theorem}

\begin{proof}
  We denote by $x_1,...,x_{2(m+n-2)}$ the points on the convex hull of $\mathcal P_{m,n}$, respecting their order.

  Let $\mathcal L$ be a set of half-planes identifying $\mathcal P_{m,n}$.
  For any $i\in\{1,...,2(m+n-2)\}$, $x_i$ and $x_{i+1}$ are separated by $\mathcal L$ (with $x_{2(m+n-2)+1}$ associated with $x_1$). It means that the there is a half-plane $L\in \mathcal L$ whose boundary line intersects  the convex hull of $\mathcal P_{m,n}$ between $x_i$ and $x_{i+1}$. In total, there are at least $2(m+n-2)$ intersections between the convex hull of $\mathcal{P}_{m,n}$ and boundary lines of some half-planes of $\mathcal L$. Since a line intersects a convex polygon into at most two points, we necessarily have $|\mathcal L|\geq m+n-2$.

 
Consider any $m-1$ vertical lines between adjacent points and $n-1$ horizontal lines between adjacent points (see Figure~\ref{fig:lines_grid} for an example). Then every pair of points is separated by a line. To obtain a solution, one just need to choose half-planes with these lines as boundary and in such a way that every point is covered.  
\end{proof}

\begin{figure}

\begin{center}
\scalebox{1.2}{
\begin{tikzpicture}

\draw (0.5,-0.5) -- (0.5,2.5);
\draw (1.5,-0.5) -- (1.5,2.5);
\draw (0.5,-0.5) -- (0.5,2.5);
\draw (2.5,-0.5) -- (2.5,2.5);
\draw (-0.5,0.5) -- (3.5,0.5);
\draw (-0.5,1.5) -- (3.5,1.5);

\foreach \pos in {(0,0),(0,1),(0,2),(1,0),(1,1),(1,2),(2,0),(2,1),(2,2),(3,0),(3,1),(3,2)}{
\fill[color=white] \pos circle (5pt) ;
\node[noeud] at \pos {};
}
\end{tikzpicture}
}
\end{center}
\caption{Identification of the edges of a grid using half-planes}
\label{fig:lines_grid}
\end{figure}
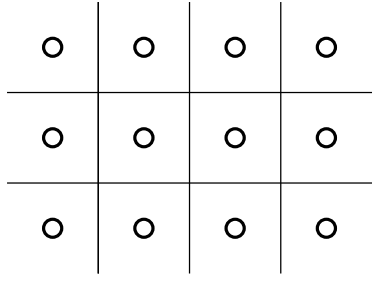

This theorem gives a bound for the general case: $\id_{D}(\mathcal P_{m,n})\leq n+m-2$. This bound is not tight, especially when $n$ is large enough compared to $m$. Next theorem gives a better (but still not tight) bound in this case:

\begin{theorem}
  Let $n$ and $m$ be two integers such that $m \geq 3$ and $n \geq \frac{m^2}{2} -3$.
 Then $\id_{D}(\mathcal P_{n,m}) \leq \lceil \frac{n}{2} \rceil + m-1$.
\end{theorem}

\begin{proof}
The idea is to use a method similar to the one described in Figure~\ref{fig:2lines}. We use half-planes to separates the lines and disks to separates the columns. When $n$ is large enough the disks act on each line in the same way they act when there is only one line.

Since $m \geq 3$ we use half-planes to separate the lines and include all the points into a disk: the bottom half-plane includes all the points above itself and the top half-plane includes all the points below itself. 

We now use disks of radius $\sqrt{\left(\frac{1}{2}\lceil \frac{n}{2} \rceil\right)^2 + \frac{m^2}{4}}$ and centered on $(\frac{1}{2}(\lceil \frac{n}{2} \rceil + 1) + k ,m/2)$ with $k$ an integer between 0 and $\lceil \frac{n}{2} \rceil - 1$. Since $n \geq \frac{m^2}{2} -3$, those disks contains $\lceil \frac{n}{2} \rceil$ points on each line and they separates all the columns. An example of such disks can be seen in Figure~\ref{fig:grille_nm}. 
Since all the points are inside a half-plane, there is no need for all the columns to be inside a disk. That is why we only need $\lceil \frac{n}{2} \rceil$ disks instead of $\lceil \frac{n+1}{2} \rceil$ as in the case of one line.

There is $\lceil \frac{n}{2} \rceil$ disks to separate the columns and $m-1$ half-planes to separate the lines, this gives us $\lceil \frac{n}{2} \rceil + m -1$ in total.

\begin{figure}[ht]
\begin{center}
\scalebox{0.8}{
\begin{tikzpicture}
    \foreach \x in {1,...,10}
        \foreach \y in {0,...,3} 
            {\node[noeud] at (\x,\y) {};}
    \foreach \x in {3,...,7}
        {\draw[very thick] (\x,1.5) circle (2.7);}
    \foreach \y in {0.5,...,2.5}
        {\draw[very thick] (-0.5,\y) -- (11,\y);}
\end{tikzpicture}
}
\end{center}
\caption{Example of identification of points of a grid with $n$ sufficiently greater than $m$}
\label{fig:grille_nm}
\end{figure}
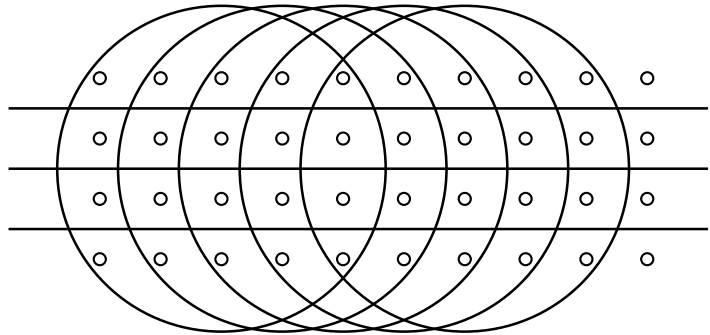

\end{proof}

\section{Extremal cases}\label{sec:bounds}
In this section, we give tight lower and upper bounds on $\id_{D}(\mathcal P)$ using the number of points of $\mathcal P$ when the points are in general configuration.

\subsection{Lower bound}

The logarithmic lower bound  given in Lemma \ref{lem:logbound} is the natural lower bound for identifying codes in hypergraphs. It is tight if any hyperedge is allowed. But if there is some structure on the hyperedges, this is not always true. In particular, if the hyperedge set has bounded dual VC-dimension $d^*$, then the lower bound is at least of order $n^{1/{d^*}}$ \cite{BDFHMP16}. This is the case for our problem since the hypergraph induced by disks have bounded dual VC-dimension equals to 3, leading to a lower bound of order $n^{1/3}$.
However, this bound is still not tight. Indeed, we provide in this section a lower bound of order $n^{1/2}$. This bound comes from the fact that an arrangment of $k$ disks can create at most $k^2-k+1$ inner faces. This classical result can be proved by induction with the argument that each time one add a circle to a set of circles it cross each circle at most twice (see \cite{oeisA014206} for more details and references). Since if some disks are identifying a set of points, there is at most one point in each faces of the intersections of the disks, we have the following bound:

\begin{proposition}\label{prop:lowerbound}
Let $\mathcal{P}$ be a set of $n$ points of $\mathbb R^2$. Then,
$\id_{D}(\mathcal P) \geq \left\lceil \frac{1+\sqrt{1+4(n-1)}}{2} \right\rceil.$ This bound is tight.
\end{proposition}

To obtain a set $\mathcal P$ of $n$ points reaching the bound, one can use an arrangment of $k$ disks making $k^2-k+1$ inner faces and set one point in each face. Such an arrangment can be obtained by taking disks of radius $1+\epsilon$, centered on vertices of a $k$-regular polygon that is inscribed in a circle of radius 1. See Figure~\ref{zone_max} for a construction with k=5.

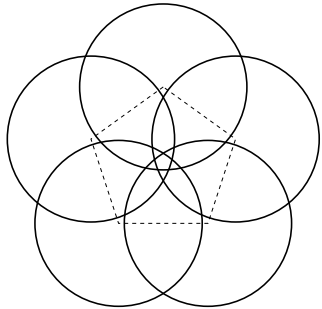
\begin{figure}[ht]
\centering
\scalebox{0.5}{
    \begin{tikzpicture}
      \draw[very thick](90:2) circle (2.2cm);
      \draw[very thick](162:2) circle (2.2cm);
      \draw[very thick](234:2) circle (2.2cm);
      \draw[very thick](306:2) circle (2.2cm);
      \draw[very thick](18:2) circle (2.2cm);
      \draw[dashed] (90:2)--(162:2)--(234:2)--(306:2)--(18:2)--(90:2);
    \end{tikzpicture}
    }
\caption{Five disks creating $21=5^2-5+1$ inner faces.}
\label{zone_max}
\end{figure}

Since an identifying code of a unit disk graph can be seen as an identifying set of special disks, the lower bound of Property~\ref{prop:lowerbound} is still true for identifying codes in unit disk graphs, improving the bound given in~\cite{BLLPT15}. Moreover, this bound is also tight for this case. Indeed, the construction of Property \ref{prop:lowerbound} can be adapted for identifying codes of unit disk graphs since all the disks have the same radius and their center can be points.

\subsection{Upper bound}
We now consider the worse configurations of points. Otherwise said, what is the minimum number of disks that is enough to identify any set of $n$ points? This question has already been solved by Gerbner and T\'oth when one just wants to separate points \cite{GT13}. They prove that $\lceil n/2 \rceil$ disks are always enough and that this is the best value one can obtain since there are point sets needing this number of disks. Since one more disk is enough to obtain an identifying set, it gives us the bound $\lceil n/2 \rceil+1$. Actually, we can slightly improve it by noticing that in the proof of Gerbner and T\'oth~\cite{GT13}, all the points are covered if there is an odd number of points. For the sake of completeness we give the proof, it follows the one of \cite{GT13}.

\begin{proposition}\label{thm:upperbound}
Let $\mathcal{P}$ be a set of $n$ points of $\mathbb R^2$. Then,
$\id_{D}(\mathcal P) \leq \lceil \frac{n+1}{2} \rceil$. This bound is tight.
\end{proposition}

\begin{proof}

  Since $\mathcal{P}$ is finite there exists a direction we can choose as abscissa, such that no pair of points have the same abscissa. Then there is a line $L$ of constant abscissa which separates the points in two parts of the same size up to one. Let $D$ be the half-plane defined by $L$ and containing the biggest part of points as one of the disks. We choose a direction perpendicular to the abscissa as the ordinate.

  At first, all the points are part of a set $P'$ and the set of identifying disks contains only $D$. Then we repeat the following operation $\lfloor \frac{n}{2} \rfloor$ times. Consider the convex hull of $P'$, exactly two of its edges cross $L$. Let $(x,y)$ be the edge which intersects $L$ on the largest ordinate. Since it is an edge of the convex hull, there is no point of $P'$ with a larger ordinate than $x$ and $y$. Therefore there is a disk $D_{x,y}$ that contains only $x$ and $y$ among the points of $\mathcal{P}'$. Add this disk to the set of identifying disks and remove $x$ and $y$ from $P'$. Iterate the process.

This algorithm gives a set of disks that identifies $P$ of size $1+\lfloor \frac{n}{2} \rfloor=\lceil \frac{n+1}{2} \rceil$. Indeed, at each step, $x$ and $y$ are separated from all the other points of $P'$ by $D_{x,y}$ and since all the other points have been considered previously they are also separated from them. Moreover, since $(x,y)$ is an edge that crosses $L$, $x$ and $y$ are separated from each other by $D$. At the end, if there is an odd number of points, there is a point that is only in $D$, since this is the only point that is only in $D$, it is identified.

Figure~\ref{fig:bh_gen} illustrates some steps of the algorithm. This bound is tight when all the points are colinear (see Theorem~\ref{th:align}). \end{proof}

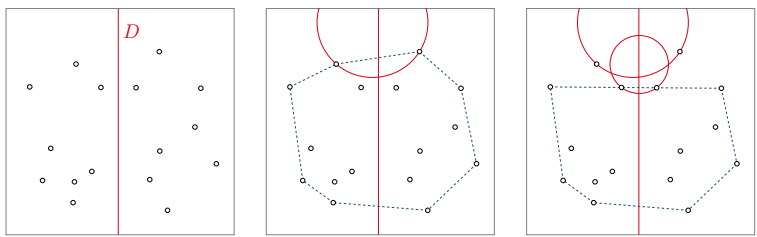
\begin{figure}[ht]%
\centering
\scalebox{0.3}{
\begin{tikzpicture} 

\draw[very thin, color = gray](0.,-2.) rectangle (10.,8.);
\clip(0.,-2.) rectangle (10.,8.);
\draw [very thick,color=rougejoli] (4.92,-2.) -- (4.92,8.);
\draw [color = rougejoli, very thick] (5,7)  node[right]{\Huge $D$};

\foreach \pos in {(1.96,1.82),(8.28,2.76),(9.22,1.14),(7.08,-0.92),(2.94,-0.58),(1.6,0.4),(3.76,0.8),(6.74,1.7),(6.3,0.44),(3.,0.34),(1.04,4.535),(8.54,4.475),(4.15996032254,4.51004031742),(5.70002175861,4.49771982593),(3.07,5.545),(6.72,6.095)}{
\node[noeud] at \pos {};}

\end{tikzpicture}
}
~
\scalebox{0.3}{
\begin{tikzpicture}
\draw[very thin, color = gray](0.,-2.) rectangle (10.,8.);
\clip(0.,-2.) rectangle (10.,8.);
\draw [very thick,color=rougejoli] (4.92,-2.) -- (4.92,8.);

\draw [dashed,very thick,color=bleujoli] (1.6,0.4)-- (1.04,4.535);

\draw [dashed,very thick,color=bleujoli] (1.04,4.535)-- (3.07,5.545);
\draw [dashed,very thick,color=bleujoli] (3.07,5.545)-- (6.72,6.095);
\draw [dashed,very thick,color=bleujoli] (6.72,6.095)-- (8.54,4.475);

\draw [very thick,color=rougejoli] (4.65939150769,7.38358363079) circle (2.4303405383cm);

\draw [dashed,very thick,color=bleujoli] (8.54,4.475)-- (9.22,1.14);

\draw [dashed,very thick,color=bleujoli] (9.22,1.14)-- (7.08,-0.92);

\draw [dashed,very thick,color=bleujoli] (7.08,-0.92)-- (2.94,-0.58);

\draw [dashed,very thick,color=bleujoli] (2.94,-0.58)-- (1.6,0.4);

\foreach \pos in {(1.96,1.82),(8.28,2.76),(9.22,1.14),(7.08,-0.92),(2.94,-0.58),(1.6,0.4),(3.76,0.8),(6.74,1.7),(6.3,0.44),(3.,0.34),(1.04,4.535),(8.54,4.475),(4.15996032254,4.51004031742),(5.70002175861,4.49771982593),(3.07,5.545),(6.72,6.095)}{
\node[noeud] at \pos {};}

\end{tikzpicture}
}
~
\scalebox{0.3}{
\begin{tikzpicture}
\draw[very thin, color = gray](0.,-2.) rectangle (10.,8.);
\clip(0.,-2.) rectangle (10.,8.);
\draw [very thick,color=rougejoli] (4.92,-2.) -- (4.92,8.);

\draw [dashed,very thick,color=bleujoli] (1.6,0.4)-- (1.04,4.535);

\draw [dashed,very thick,color=bleujoli] (8.54,4.475)-- (9.22,1.14);

\draw [dashed,very thick,color=bleujoli] (9.22,1.14)-- (7.08,-0.92);

\draw [dashed,very thick,color=bleujoli] (7.08,-0.92)-- (2.94,-0.58);

\draw [dashed,very thick,color=bleujoli] (2.94,-0.58)-- (1.6,0.4);

\draw [color=rougejoli] (4.65939150769,7.38358363079) circle (2.4303405383cm);

\draw [dashed,very thick,color=bleujoli] (1.04,4.535)-- (8.54,4.475);

\draw [very thick,color=rougejoli] (4.93817791558,5.5272394479) circle (1.28074848943cm);

\foreach \pos in {(1.96,1.82),(8.28,2.76),(9.22,1.14),(7.08,-0.92),(2.94,-0.58),(1.6,0.4),(3.76,0.8),(6.74,1.7),(6.3,0.44),(3.,0.34),(1.04,4.535),(8.54,4.475),(4.15996032254,4.51004031742),(5.70002175861,4.49771982593),(3.07,5.545),(6.72,6.095)}{
\node[noeud] at \pos {};}

\end{tikzpicture}
}
\caption{First steps of the algorithm for the general upper bound}%
\label{fig:bh_gen}%
\end{figure}

All the values between the lower bound of Property \ref{prop:lowerbound} and the upper bound of Property \ref{thm:upperbound} are reached:

\begin{theorem}
Let $n \in \mathbb{N}$ and $k \in \mathbb{N}$ be such that $\lceil \frac{1+\sqrt{1+4(n-1)}}{2} \rceil \leq k \leq \lceil \frac{n+1}{2} \rceil$.

There exists an $n$-point set $\mathcal{P}$ of $\mathbb R^2$ such that $\id_{D}(\mathcal P) = k$.
\end{theorem}

\begin{proof}
  Consider the optimal arrangement of $k$ disks based on a regular polygon given on Figure~\ref{zone_max}. There is a line $L$ cutting this construction into $2k - 1$ regions. Indeed, let $L'$ be a line going through an intersection of disks and the center of the polygon, for symmetry reason, this line goes through $k$ regions and $k-1$ intersections. So by shifting the line infinitesimally and parrallely, it is still going through the previous regions but for each intersection we have a new region. Therefore, there is indeed $2k-1$ regions crossed by this new line $L$.
  We set $2k-1$ points on $L$ in the different regions of the arrangement of disks. Since $k \geq \lceil \frac{1+\sqrt{1+4(n-1)}}{2} \rceil$ we can set the $n-(2k-1)$ remaining points in the other regions of the arrangement. At the end, the set of $k$ disks of the arrangment identifies the $n$ points that we have put in its different regions and, since there are $2k-1$ colinear points, no smaller set of disks can identify these points.

\end{proof}

\subsection{Improved upper bound for general configurations}

The upper bound of Proposition \ref{thm:upperbound} is tight for colinear points. A natural question is whether the bound is still tight if there are no three colinear points among $\mathcal P$. Actually, the bound is also tight if the points are cocyclic. But, if there are no three colinear points nor four cocyclic points in $\mathcal P$, the upper bound is not tight anymore. In this section, we said that a set of points of $\mathbb R^2$ is in {\em general configuration} if there are no three points of $\mathcal P$ on a line nor four points of $\mathcal P$ on a circle.

\begin{theorem}
\label{th:bh_frequente}
Let $\mathcal{P}\subseteq \mathbb R^2$ be a set of $n$ points in general configuration. Then  $\id_{D}(\mathcal P) \leq 2\lceil n/6 \rceil +1.$
\end{theorem}

The idea of the proof of this theorem is to give an algorithm that constructs an identifying set of disks of size $2\lceil n/6 \rceil +1$. The algorithm is based on the same principle that we used in the not restricted case:

\begin{enumerate}
\item Divide $\mathcal{P}$ in three equal parts using lines;
\item Choose a disk that contains exactly one point in each part, remove these points and repeat the operation.
  \end{enumerate}

The crucial part is to find the disk of the second step. For that, we use Delaunay triangulations - that is a triangulation of the points in such a way that the circumcircle of each triangle only contains its vertices. Since there are no three colinear nor four cocyclic points, such a triangulation always exists (and is unique). To find the disk of Step 2, we then need to find a Delaunay triangle that has a vertex in each part, as illustrated in Figure~\ref{fig:bh_frequente}. To insure the existence of such a triangle, we need to be more precise at Step~1.

%

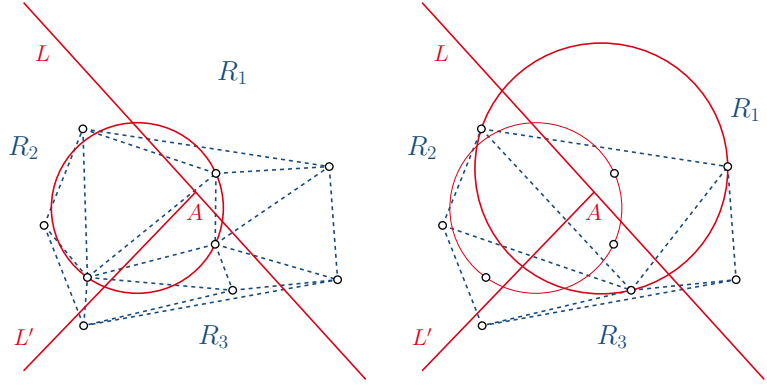
\begin{figure}[ht]%
\centering
\scalebox{0.5}{
\begin{tikzpicture}

\draw [very thick,color=bleujoli,dashed] (5.58,-0.72)-- (5.68,0.56);
\draw [very thick,color=bleujoli,dashed] (4.54,1.94)-- (5.68,0.56);
\draw [very thick,color=bleujoli,dashed] (4.54,1.94)-- (5.58,-0.72);
\draw [very thick,color=bleujoli,dashed] (9.04,1.44)-- (9.5,0.22);
\draw [very thick,color=bleujoli,dashed] (9.06,3.32)-- (9.04,1.44);
\draw [very thick,color=bleujoli,dashed] (9.06,3.32)-- (12.04,3.5);
\draw [very thick,color=bleujoli,dashed] (5.56,4.5)-- (12.04,3.5);
\draw [very thick,color=bleujoli,dashed] (12.04,3.5)-- (12.26,0.5);
\draw [very thick,color=bleujoli,dashed] (5.58,-0.72)-- (12.26,0.5);
\draw [very thick,color=bleujoli,dashed] (4.54,1.94)-- (5.56,4.5);
\draw [very thick,color=bleujoli,dashed] (9.5,0.22)-- (12.26,0.5);
\draw [very thick,color=bleujoli,dashed] (9.5,0.22)-- (5.58,-0.72);
\draw [very thick,color=bleujoli,dashed] (9.04,1.44)-- (12.04,3.5);
\draw [very thick,color=bleujoli,dashed] (9.04,1.44)-- (12.26,0.5);
\draw [very thick,color=bleujoli,dashed] (5.68,0.56)-- (9.04,1.44);
\draw [very thick,color=bleujoli,dashed] (5.56,4.5)-- (9.06,3.32);
\draw [very thick,color=bleujoli,dashed] (5.56,4.5)-- (5.68,0.56);
\draw [very thick,color=bleujoli,dashed] (5.68,0.56)-- (9.06,3.32);
\draw [very thick,color=bleujoli,dashed] (5.68,0.56)-- (9.5,0.22);

\draw [very thick,color=rougejoli,domain=4.:13.] plot(\x,{(--73.4988-6.64*\x)/5.98});
\draw [very thick,color=rougejoli,domain=4.0:8.526177482780211] plot(\x,{(--22.9895381801-3.94357550407*\x)/-3.76617748278});
\draw [very thick,color=rougejoli] (4,-1) node {\LARGE{$L'$}};
\draw [very thick,color=rougejoli] (4.5,6.5) node {\LARGE{$L$}};
\draw [very thick,color=rougejoli] (8.5,2.3) node {\LARGE{$A$}};

\draw [very thick, color = bleujoli] (9.5,6) node {\huge{$R_1$}};
\draw [very thick, color = bleujoli] (4,4) node {\huge{$R_2$}};
\draw [very thick, color = bleujoli] (9,-1) node {\huge{$R_3$}};

\draw [very thick,color=rougejoli] (6.99283972568,2.40188468377) circle (2.26187694927cm);

\foreach \pos in {(5.56,4.5),(4.54,1.94),(9.06,3.32),(9.04,1.44),(5.68,0.56),(5.58,-0.72),(12.04,3.5),(9.5,0.22),(12.26,0.5)}{
\node[noeud] at \pos {};}

\end{tikzpicture}
}
~
\scalebox{0.5}{
\begin{tikzpicture}

\draw [very thick,color=bleujoli,dashed] (4.54,1.94)-- (5.58,-0.72);
\draw [very thick,color=bleujoli,dashed] (5.56,4.5)-- (12.04,3.5);
\draw [very thick,color=bleujoli,dashed] (12.04,3.5)-- (12.26,0.5);
\draw [very thick,color=bleujoli,dashed] (5.58,-0.72)-- (12.26,0.5);
\draw [very thick,color=bleujoli,dashed] (4.54,1.94)-- (5.56,4.5);
\draw [very thick,color=bleujoli,dashed] (9.5,0.22)-- (12.26,0.5);
\draw [very thick,color=bleujoli,dashed] (9.5,0.22)-- (5.58,-0.72);
\draw [very thick,color=bleujoli,dashed] (4.54,1.94)-- (9.5,0.22);
\draw [very thick,color=bleujoli,dashed] (9.5,0.22)-- (12.04,3.5);
\draw [very thick,color=bleujoli,dashed] (5.56,4.5)-- (9.5,0.22);

\draw [very thick,color=rougejoli,domain=4.:13.] plot(\x,{(--73.4988-6.64*\x)/5.98});
\draw [very thick,color=rougejoli,domain=4.0:8.526177482780211] plot(\x,{(--22.9895381801-3.94357550407*\x)/-3.76617748278});
\draw [very thick,color=rougejoli] (4,-1) node {\LARGE{$L'$}};
\draw [very thick,color=rougejoli] (4.5,6.5) node {\LARGE{$L$}};
\draw [very thick,color=rougejoli] (8.5,2.3) node {\LARGE{$A$}};

\draw [very thick, color = bleujoli] (12.5,5) node {\huge{$R_1$}};
\draw [very thick, color = bleujoli] (4,4) node {\huge{$R_2$}};
\draw [very thick, color = bleujoli] (9,-1) node {\huge{$R_3$}};

\draw [thin,color=rougejoli] (6.99283972568,2.40188468377) circle (2.26187694927cm);
\draw [very thick,color=rougejoli] (8.71529939818,3.45114010019) circle (3.32505960572cm);

\foreach \pos in {(5.56,4.5),(4.54,1.94),(9.06,3.32),(9.04,1.44),(5.68,0.56),(5.58,-0.72),(12.04,3.5),(9.5,0.22),(12.26,0.5)}{
\node[noeud] at \pos {};}

\end{tikzpicture}
}

\caption{First steps of the method used in the main idea of the proof}%
\label{fig:bh_frequente}%
\end{figure}

Before going into details, we need two preliminary results.

\begin{theorem}[Ceder \cite{ceder1964generalized}]
\label{lem:sixpartite}
For $n$ points of $\mathbb R^2$ with no three colinear points, there is a way to divide the plan in six regions containing each between $\lceil \frac{n}{6} \rceil -1$ and $\lceil \frac{n}{6} \rceil$ points using three concurrent lines.
\end{theorem}

\begin{lemma}
\label{lem:intersection}
Let $\mathcal{P}$ be a set of points of $\mathbb R^2$, $L$ a line and $L'$ a half-line with origin $A$ on $L$. If each of the three regions $R_1$, $R_2$ and $R_3$ made by $L$ and $L'$ contains one point of $\mathcal{P}$ and if $A$ is in the convex hull of $\mathcal{P}$, then every triangulation of $\mathcal{P}$ contains a triangle that has a vertex in each region.
\end{lemma}

\begin{proof}
  
  Let $T$ be a triangulation of $\mathcal{P}$. Since the intersection $A$ of $L$ and $L'$ is inside the convex hull, there is at least a segment of $T$ between any pair of regions.

  Consider the segments between the regions separated by $L'$, namely $R_2$ and $R_3$. Let $[x,y]$ be the segment which cut $L'$ the closest from $A$. Since $A$ is in the convex hull of $\mathcal{P}$, there is at least one point $z$ of $\mathcal{P}$ such that $(x,y,z)$ is a triangle of $T$ in the direction of $A$ from this segment. If $z$ is in $R_2$ or $R_3$, then the segment $[x,z]$ or $[y,z]$ would intersects $L'$ closer to $A$ than $[x,y]$, which contradicts the hypothesis that $[x,y]$ is the closest segment to $A$. Therefore $z$ is in $R_1$ and $(x,y,z)$ form a triangle with one vertex in each region.

\end{proof}

\begin{proof}[Proof of Theorem \ref{th:bh_frequente}]

Using Theorem~\ref{lem:sixpartite}, there exist three concurrent lines $L_1$, $L_2$ and $L_3$ that divides the plane into six regions of the same size up to one. Let $A$ be their common intersection. Let $D_1$, $D_2$ and $D_3$ be three half-planes defining by $L_1$, $L_2$ and $L_3$ such that every point is in at least one half-plane.  Let $a$, $b$, $c$, $d$, $e$ and $f$ be the six regions of the plane created by these lines, as illustrated in Figure~\ref{fig:sixpartite}.
    
First consider the regions $a$, $c$ and $e$. Each of these regions contains between $\lceil \frac{n}{6} \rceil - 1$ and $\lceil \frac{n}{6} \rceil$ points. For construction reasons, if there is a point in each region then $A$ is inside the triangle formed by these points, so Lemma~\ref{lem:intersection} always applies.

Consider the following process. Add $D_1$, $D_2$ and $D_3$ to the future set of identifying disks. Set in $P'$ all the points of $a$, $c$ and $e$.
Then repeat the following operation $\lceil \frac{n}{6} \rceil - 1$ times. Consider the Delaunay triangulation of $P'$, at least one triangle $(x,y,z)$ has its vertices in the three different regions. Since it is a triangle of a Delaunay triangulation, its circumscribed circle $C$ contains no other remaining points. Add $D_{x,y,z}$, the disk of perimeter $C$, to the set of identifying disks, remove $x$, $y$ and $z$ from $P'$ and iterate the process.

We do the same iterated operation for the regions $b$, $d$ and $f$.

At the end of each step of the process all the considered points are separated from all the other points. Indeed, at each step $x$, $y$ and $z$ are separated from the points of the other regions by the half-planes and, since each considered triangle comes from a Delaunay triangulation, their circumscribed circle contains no other points of $P'$. 

If a point has not been considered at the end of the processes, it is alone in its regions and therefore is isolated from all the other points. Moreover, by the selection of the half-planes, it is inside at least one half-plane so it is covered.

Therefore, this algorithm constructs an identifying set of disks of size $3+2(\lceil \frac{n}{6} \rceil -1) = 2\lceil \frac{n}{6} \rceil +1$.
\end{proof}

\begin{figure}[ht]
\centering
\begin{tikzpicture}

\newcommand{\E}{(0,6) .. controls +(1,0) and +(-1,1) .. (3,4) 
            .. controls +(1,-1) and +(2,2) .. (6,-3)
            .. controls +(-2,-2) and +(1,-2) .. (-3,-2)
            .. controls +(-1,2) and +(-1,-1) .. (-3,3)
            .. controls +(1,1) and +(-1,0) .. (0,6)}
\newcommand{\F}{(-7,0)--(7,0)--(7,5)--(-7,5)--cycle}
\newcommand{\G}{(-5,-5)--(7,5)--(5,-5)--cycle}
\newcommand{\K}{(-3,6)--(5,-6)--(-5,-6)--cycle}

  \begin{scope}[scale=0.3]
\begin{scope}[color = lightgray, opacity = 0.5]
\clip \E;
\begin{scope}
\clip \F;
\fill \G;
\fill \K;
\end{scope}
\begin{scope}
\clip \G;
\fill \K;
\end{scope}
\end{scope}
\draw \E;  
\draw[very thick,color=rougejoli] (-7,0)--(7,0);
\draw[very thick,color=rougejoli] (-5,-5)--(7,5);
\draw[very thick,color=rougejoli] (-3,6)--(5,-6);
\node at (1,3) {{$\frac{1}{6}$}};
\node at (3,5.5) {{$a$}};
\node at (4,-2) {{$\frac{1}{6}$}};
\node at (7,2) {{$f$}};
\node at (-2,-1) {{$\frac{1}{6}$}};
\node at (8,-3) {{$e$}};
\node at (1,-2.5) {{$\frac{1}{6}$}};
\node at (1,-5.5) {{$d$}};
\node at (4,1) {{$\frac{1}{6}$}};
\node at (-4.5,-2) {{$c$}};
\node at (-2,1.5) {{$\frac{1}{6}$}};
\node at (-5,2) {{$b$}};

\draw[very thick,color=rougejoli] (5.8,-5.2) node {$L_1$};
\draw[very thick,color=rougejoli] (-4,-5.3) node {$L_2$};
\draw[very thick,color=rougejoli] (-6.3,-0.8) node {$L_3$};
\end{scope}

\end{tikzpicture}
\caption{Dividing the points into six equal parts}
\label{fig:sixpartite}
\end{figure}
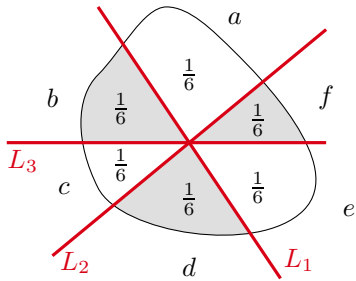

The previous bound is tight, up to a constant 2, when points are located on an half-parabola, the curve constituted of one side of a parabola symmetry axis:

\begin{proposition}
Let $\mathcal{P}\subseteq \mathbb R^2$ be a set of $n$ points placed on an half-parabola. Then, $\id_{D}(\mathcal P) \geq \frac{n}{3}.$
\end{proposition}

\begin{proof} 

  Let $\mathcal{P}$ be a set of $n$ points located on a half-parabola $H$. We denote by $x_1,...,x_n$ the points, respecting their order on $H$, with $x_1$ the closest to the extrema of $H$.

  Let $\mathcal D$ be a set of disks identifying $\mathcal P$.
  For any $i\in\{1,...,n-1\}$, $x_i$ and $x_{i+1}$ are separated by $\mathcal D$. It means that there is a disk $D\in \mathcal D$ whose perimeter intersects $H$ between $x_i$ and $x_{i+1}$. Moreover, $x_n$ is covered by $\mathcal D$, thus there is a disk whose perimeter intersects $H$ after $x_n$. In total, there are at least $n$ intersections between $H$ and some perimeters of disks of $\mathcal D$.

We now prove that a circle $C$ can intersect $H$ into at most three points.
Let $(x,y) \in C \cap H$. Without loss of generality, $(x,y)$ satisfies the following set of equations, with $x_0$, $y_0$, $r$ that are constant.

  \[
  \left\{
  \begin{array}{l} 
        y = x^2\\
        (x-x_0)^2+(y-y_0)^2 = r^2\\
        x\geq 0
  \end{array}
    \right.
    \]

In particular, $x$ is a solution of:

\begin{equation}
(X-x_0)^2+(X^2-y_0)^2 = r^2
\label{eq:halfpar}
\end{equation}

There is no term in $X^3$ in the previous equation. Thus, if $x_1$, $x_2$, $x_3$ and $x_4$ are solutions of (\ref{eq:halfpar}, we have $x_1 + x_2 + x_3 + x_4 = 0$. Since $x \geq 0$, there are at most three possible values for $x$.

Since there are at least $n$ intersections between $H$ and an identifying set of disks $\mathcal{D}$ and since a circle intersects a half-parabola at most three times, we necessarly have $\mathcal{D} \geq \frac{n}{3}$.

\end{proof}

\section{Complexity when the radius is fixed}\label{sec:complexity}

In this section, we consider the complexity of the following decision problems (with $r \in \mathbb R$):

\noindent\decisionpb{{\sc Identification-Disk}(r)}{A finite set $\mathcal P\subseteq R^2$, an integer $k$.}{Is it true that $\id_{D,r}(\mathcal P)\leq k$?}{0.9}

\noindent\decisionpb{{\sc Colinear Identification-Disk}(r)}{A finite set $\mathcal P\subseteq R^2$ of colinear points, an integer $k$.}{Is it true that $\id_{D,r}(\mathcal P)\leq k$?}{0.9}

\begin{theorem}
  {\sc Identification-Disk}(r) is $\mathcal{NP}$-complete.
\end{theorem}

\begin{proof}
  We prove the result for $r=1/2$ which is not restrictive. We reduce this problem from the problem of partitioning a grid graph into path on three vertices. A {\em grid graph} is a graph with vertex set included in $\mathbb Z^2$ and two vertices are adjacent if they are at Euclidean distance 1.

\noindent\decisionpb{$P_3$-{\sc Partition-Grid}}{A grid graph $G$.}{Is there a partition of the vertices of $G$ in such a way that each part induces a path on three vertices?}{0.9}

Bevern {\em et al.} \cite{P3}  proved that this problem is ${\mathcal NP}$-complete.

\

Let $G$ be an instance of $P_3$-{\sc Partition-Grid} and $n=|V(G)|$. The instance of the identification problem is $\mathcal{P} = V(G)$ and $k=2n/3$. We have to prove that $G$ has a $P_3$-partition if and only if $V(G)$ is identified by $2n/3$ disks of radius $1/2$.

\smallskip

Assume first that there is a partition of $G$ into path on three vertices. For each part $(x,y,z)$ of the $P_3$ partition, add to the identifying set the disk $D_{x,y}$ of radius $1/2$ that contains $x$ and $y$ and the disk $D_{y,z}$ that contains $y$ and $z$. Since $(x,y)$ and $(y,z)$ are edges of $G$, these disks exist and contains exactly two points. Furthermore, $x$ is the only point that is contained only in $D_{x,y}$, $z$ is the only point that is contained only in $D_{y,z}$ and $y$ is the only point that is contained exactly in $D_{x,y}$ and $D_{y,z}$. Hence, we obtain an identifying set of disks of size $2n/3$.

\smallskip

Assume now that there is an identifying set of disks $\mathcal D$ that identify $V(G)$ with $2n/3$ disks. Since the points are at distance at least 1, every disk contains at most two points. Without loss of generality, we can suppose that if a disk contains only one point, then this point is not included in any other disk. Indeed, assume there are two points $x,y$ and two disks $D_1$ and $D_2$ such that $D_1$ contains only $x$ and $D_2$ contains both $x$ and $y$, then we can replace $D_2$ by $D_2'$ that contains only $y$ and the situation is similar, $V(G)$ is still identified and the number of disks is the same.

Let $a$ be the number of disks containing only one point. Let $V'$ be the $n-a$ points not covered by these $a$ disks. Let $G'$ be the graph with vertices $V'$ and edges $(x,y)$ if $x,y$ are contained together in a disk of $\mathcal D$. Note that $G'$ is a subgraph of $G$. The graph $G'$ has $n-a$ vertices and its connected components have at least three vertices. Indeed, if a component as only two vertices then these vertices are not identified. So there are $k \leq (n-a)/3$ connected components. We name these connected components $\{G_1, ... ,G_k\}$. The number of edges of $G'$ is $\sum_{i=1}^{k}E(G_i) \geq \sum_{i=1}^{k}(V(G_i) -1) \geq (n-a) -k \geq 2(n-a)/3$.

Since a disk is either containing one point (and there are $a$ such disks) or corresponds to an edge of $G'$, there are at least $a+2(n-a)/3 = 2n/3+2a/3$ disks in $\mathcal D$.  Therefore we necessarily have $a=0$ and there are exactly $n/3$ connected components in $G'$, each of them being of size 3. This is a $P_3$-partition of $G$.
\end{proof}

However if all of the points are colinear then this problem can be solved in linear time:

\begin{theorem}
\label{th:colin_id_r}
	{\sc Colinear Identification-Disk}(r) can be solved in linear time.
\end{theorem}

Note that if the disks are required to be centered on the points, this problem is equivalent to the problem of identifying codes in unit interval graphs, whose complexity is surprisingly still open.

\medskip

To prove Theorem~\ref{th:colin_id_r}, we introduce few definitions and preliminary results.
Let $\mathcal{P}$ be a set of $n$ colinear points on a line $L$ and $\mathcal{D}$ is a set of disks of radius $r$ identifying $\mathcal P$. Let $x_1$,..$x_n$ be the points of $\mathcal P$. We confuse $x_i$ with its abscissa on $L$ and we assume that $x_1<...<x_n$.

Note that, since the points are colinear and the centers of the points can be chose anywhere on the plane, a set of points of $\mathcal P$ can be the intersection of $\mathcal P$ and a disk of radius $r$ if and only if there are consecutive points of $\mathcal P$ at distance at most $2r$. In the following, we will refer often directly to the set of points contained in a disk $D$ instead of $D$ itself.

The set $\mathcal{D}$ is \textit{optimal} if $|\mathcal{D}| = \id_D(\mathcal P)$, it is \textit{perfect} if $n$ is odd and if $|\mathcal{D}| = \frac{n+1}{2}$ (in particular $\mathcal{D}$ is optimal).

The disks of $\mathcal{D}$ partitions the points of $\mathcal P$ on connected components. More formally, we define an equivalence relation $x \sim_{\mathcal{D}} y$ meaning that $x$ and $y$ are connected by a path of disks. For $x$ and $y$ two points of $\mathcal{P}$, we have $x \sim_{\mathcal{D}} y$ if and only if there is a disk $D$ in $\mathcal{D}$ that contains both $x$ and $y$, or there is a point $z$ in $\mathcal{P}$ such that $x \sim_{\mathcal{D}} z$ and $y \sim_{\mathcal{D}} z$.

The equivalence classes $(\mathcal{P}_i)$ of $\sim_{\mathcal{D}}$  are made of consecutive points. Let $\mathcal{D}_i = \{D \in \mathcal{D} | D \cap \mathcal{P}_i \neq \emptyset\}$ be the disks containing points the points of $(\mathcal{P}_i)$. 

A set of disks $\mathcal{D}$ is \textit{piece-wise perfect} if it is optimal and if each $\mathcal{D}_i$ perfectly identifies $\mathcal{P}_i$.

\begin{lemma}
\label{lem:pw_perfect}
For any set of colinear points $\mathcal{P}$, there is a set $\mathcal{D}$ of disks of radius $r$ that identifies $\mathcal{P}$ and is piece-wise perfect.
\end{lemma}

\begin{proof}
Let $\mathcal{D}$ be a set of disks that identifies optimally $\mathcal{P}$ and such that $\sum\limits_{D \in \mathcal{D}}{|D \cap \mathcal{P}|}$ is minimal.
We will prove that this set is piece-wise perfect.

Assume the contrary. It means that there is a set $\mathcal{P}_i$ which is not perfectly identified by $\mathcal{D}_i$. Following the proof of Theorem~\ref{th:align}, this means that there are two disks $D_1$ and $D_2$ whose perimeters intersect $L$ between the same pair of adjacent points of $\mathcal{P}_i$ or both before the first point of $\mathcal{P}_i$ or both after the last point of $\mathcal{P}_i$.
Let $x_a,...,x_b$ be the points covered by $D_1$ and $x_c,...,x_d$ the points covered by $D_2$.

\textbf{Case 1 }: $a = c$ (the case $b = d$ is similar). Suppose, without loss of generality, that $d \leq b$. Let $D_1'$ be a disk that contains the points from $x_a+1$ to $x_{b}$, such a disk exist because its intersection with $\mathcal{P}$ is included in $D_1 \cap \mathcal{P}$. Then $\mathcal{D}' =\mathcal{D} \setminus \{D_1\} \cup \{D_1'\}$ identifies $\mathcal{P}$. Indeed, the only point of $\mathcal{P}$ for whom the situation is different for $\mathcal{D}$ and $\mathcal{D}'$ is $x_a$. For $\mathcal{D}'$ it is the only point of $\mathcal{P}$ that is inside $D_2$ and not inside $D_1'$. So we have a new set of disks that identifies $\mathcal{P}$ and such that the sum of the number of points contained in each disk is smaller. This is a contradiction to the minimal property of $\mathcal{D}$.

\textbf{Case 2} : $c = b+1$ (the case $a = d+1$ is similar).
Since $\mathcal{P}_i$ is an equivalence class for the relation $\sim_{\mathcal{D}}$, we must have $x_b \sim_{\mathcal{D}} x_{b+1}$ and there must be a disk $D_3$ such that $D_3$ contains both $x_b$ and $x_{b+1}$. Let  $x_e$ the first point of $D_3$ and $x_f$ its last point.

\textbf{Subcase 2.1} : $a<e<b<f<d$.

Let $D_1'$ be a disk that contains the point from $x_a$ to $x_{b-1}$, such a disk can exist because it is smaller than $D_1$, and let $\mathcal{D}' = \mathcal{D} \setminus \{D_1\} \cup \{D_1'\}$. $\mathcal{D}'$ identifies $\mathcal{P}$ and contradict the minimal property of $\mathcal{D}$.

\textbf{Subcase 2.2} : $e<a$ (the case $f>d$ is similar).

Let $D_1'$ be a disk that contains the point from $x_e$ to $x_b$, such a disk can exist because it is smaller than $D_3$ and $D_3'$ be a disk that contains the points from $x_a$ to $x_f$, such a disk can exist because it is smaller than $D_3$. The set of disks $\mathcal{D}' = \mathcal{D} \setminus \{D_1,D_3\} \cup \{D_1',D_3'\}$ identifies $\mathcal{P}$. Indeed, the only points of $\mathcal{P}$ for whom the situation is different for $\mathcal{D}$ and $\mathcal{D}'$ are those between $x_e$ and $x_{a-1}$. They are still separated from each other by the disks that separates them in $\mathcal{D}$ and they are separated from the other points because they are the only points that are in $D_1'$ and not in $D_3'$. The sum of the number of points contained in the disk is the same for $\mathcal{D}'$ and $ \mathcal{D}$.

Finally, the sum of the number of points contained in the disks remains the same, but now the disk that contains both $x_b$ and $x_{b+1}$ does not contain the disks that intersect the region between $x_b$ and $x_{b+1}$. So we are now in Subcase 2.1 and we can apply the method used in this case, concluding the proof.

\end{proof}

A set of disks $\mathcal{D}$ identifying a set $\mathcal{P}$ of $n$ colinear points is in \textit{normal form} if $n$ is odd and if :
\begin{itemize}
\item $n=1$ and $\mathcal{D}$ is composed of a unique disk containing the point of $\mathcal{P}$.

or

\item $n=2p+1$, $p\geq 1$, and $\mathcal{D}= \{D_i\}_{i\in [0,p]}$ with $D_0$ containing the points $x_1$ and $x_2$, $D_p$ containing the points $x_{2p}$ and $x_{2p+1}$ and, for $i\in [1,p-1]$, $D_i$ containing the points $x_{2i}$, $x_{2i+1}$ and $x_{2i+2}$.
\end{itemize}

In particular, $\mathcal{D}$ is perfect.

\begin{lemma}
\label{lem:perfectTo3}
For any set of colinear points $\mathcal{P}$, if there is a set of disks that perfectly identifies $\mathcal{P}$, then there exists a set $\mathcal{D}$ that identifies $\mathcal{P}$ and is in normal form.
\end{lemma}

\begin{proof}
If there is only one point in $\mathcal{P}$, then the only way to identify perfectly $\mathcal{P}$ is to have a set $\mathcal{D}$ that contains exactly one disk which contains the point of $\mathcal{P}$, it always exists and it is already in normal form.

Suppose that $\mathcal P$ is of size $2p+1$ with $p \geq 1$. We show that if there is no set of disks identifying $\mathcal{P}$ in normal form then there is no set perfectly identifying $\mathcal{P}$.

Assume that there is no set of disks identifying $\mathcal P$ in normal form but that a set $\mathcal D$ perfectly identifies $\mathcal P$. Necessarily, $\sim_{\mathcal D}$ has only one equivalnce class and all the adjacent points of $\mathcal P$ are distant at most $2r$.

Since there is no possible set in normal form, there exists $i \in [1,p-1]$ such that the distance between the points $x_{2i}$ and $x_{2i+2}$ is greater than $2r$. Let $\mathcal P_1$ be the set $\{x_1, ... , x_{2i}\}$ and $\mathcal P_2$ be the set $\{x_{2i+2}, ... , x_{2p+1}\}$. Let $\mathcal{D}_1$ (respectively $\mathcal D_2$) be the subset of disks of $\mathcal D$ that contains at least one point of $\mathcal P_1$ (resp. $\mathcal P_2$). The intersection between $\mathcal D_1$ and $\mathcal D_2$ is empty since the distance between $x_{2i}$ and $x_{2i+2}$ is at least $2r$. By Theorem~\ref{th:align}, since $\mathcal D_1$ identifies $\mathcal P_1$ and $\mathcal D_2$ identifies $\mathcal P_2$, $|\mathcal D_1| \geq \lceil \frac{2i+1}{2}\rceil$ and $|\mathcal D_2| \geq \lceil \frac{2(p-i)+1}{2}\rceil$. So $|\mathcal D| \geq |\mathcal D_1| +|\mathcal D_2| \geq \lceil \frac{2i+1}{2}\rceil +  \lceil \frac{2(p-i)+1}{2}\rceil = p+2$. Hence $\mathcal D$ does not identify $\mathcal P$ perfectly, a contradiction.
\end{proof}

\begin{proof}[Proof of Theorem~\ref{th:colin_id_r}]
By Lemma~\ref{lem:pw_perfect}, there exist a set identifying $\mathcal{P}$ that is piece-wise perfect. By Lemma~\ref{lem:perfectTo3}, every perfect part of that set can be identified by a set of disks in normal form. So there is a piece-wise perfect set of disks $\mathcal{D}$ such that every $\mathcal{D}_i$ is in normal form.

We now give an algorithm that finds an optimal solution to identify $\mathcal{P}$ with connected sets of disks that are in normal form :

\begin{algorithm}
\begin{algorithmic}
\REQUIRE the abscissas ${x_1,..., x_n}$ of a set of colinear points $\mathcal{P}$
\ENSURE $\mathcal{D}$ is a minimal identifying set that can be partitioned into subset in normal form.

\STATE $i \leftarrow 0$
\STATE $\mathcal{D} \leftarrow \emptyset$
\STATE $x_{n+1} \leftarrow \infty$, $x_{n+2} \leftarrow \infty$, $x_{n+3} \leftarrow \infty$
\WHILE{$i \leq n$}

\IF{$x_{i+1} - x_i > 2r$ or $x_{i+2} - x_{i+1} > 2r$}
\STATE Add to $\mathcal{D}$ a disk that contains only $x_i$
\STATE $i \leftarrow i+1$

\ELSE
\STATE Add to $\mathcal{D}$ a disk that contains only $x_i$ and $x_{i+1}$
\STATE $i \leftarrow i+1$

\WHILE{$x_{i+2}-x_i \leq 2r$ and $x_{i+3}-x_{i+2} \leq 2r$}
\STATE Add to $\mathcal{D}$ a disk that contains only $x_i$, $x_{i+1}$ and $x_{i+2}$
\STATE $i \leftarrow i+2$

\ENDWHILE

\STATE Add to $\mathcal{D}$ a disk that contains only $x_i$ and $x_{i+1}$
\STATE $x_i \leftarrow x_{i+2}$

\ENDIF
\ENDWHILE
\end{algorithmic}
\end{algorithm}

This algorithm takes the biggest connected sets of disks in normal form starting with the first point of $\mathcal{P}$. We prove that this is optimal.
Assumme it is not the case. Let $\mathcal{P}$ be a set of points such that the set $\mathcal{D}$ given by the algorithm is not an optimal solution. We choose $\mathcal P$ with a minimum number of points.
Let $\mathcal{D}^{opt}$ be an optimal set in normal form. Its first connected component $\mathcal{D}^{opt}_1$ is smaller than the first connected component of $\mathcal{D}$, $\mathcal{D}_0$. Indeed, it cannot be bigger since the algorithm take the biggest connected component and it cannot be the same since, by minimality of $\mathcal P$, the algorithm is optimal on the rest of the points. So $\mathcal{D}^{opt}$ identifies $\mathcal{P}^{opt}_1$, the points of $\mathcal{P}$ that are not in the disks of $\mathcal{D}^{opt}_1$ with less disks than $\mathcal{D}$ uses to identify $\mathcal{P}_1$, the points of $\mathcal{P}$ that are not in $\mathcal{D}'$. Since $\mathcal{P}_1 \subset \mathcal{P}^{opt}_1 $, $\mathcal{D}_{opt}$ also identifies $\mathcal{P}_1$, and thus with less disks than $\mathcal{D}$. This contradicts the minimality of $\mathcal{P}$.

This algorithm is linear since we consider each point at most once.

So there is a linear algorithm to find the maximum number of disks needed to identify a set of points if each connected part must be in normal form. By Lemma~\ref{lem:pw_perfect} and Lemma~\ref{lem:perfectTo3}, this algorithm also gives a solution to {\sc Colinear Identification-Disk}(r).
\end{proof}

\section{Conclusion}

We conclude with some open problems. About complexity issues, we do not know if computing a minimum identifying set of disks when the radius is not fixed is $\mathcal NP$-complete, but the contrary would be surprising. The complexity of identification with lines seems to be also open. An intersecting question is what is the number of disks needed if the points are randomly chosen in a $1\times1$ square. It would also be interesting to consider identifications with other sets or in higher dimensions using balls instead of disks.

\bibliographystyle{plain}
\bibliography{bibli}
\end{document}